\documentclass[11pt]{article}
\usepackage{color,soul}
\usepackage[english]{babel}
\usepackage{bbm}
\usepackage{titlesec}
\usepackage{hyperref}
\usepackage{relsize}
\usepackage{mathtools}
\usepackage{dirtytalk}
\usepackage{tabularx}
\usepackage{caption}
\usepackage{babel}
\usepackage[font=small,labelfont=bf,tableposition=top]{caption}
\usepackage{booktabs}
\usepackage{rotating}
\usepackage{breqn}

\usepackage{graphicx}
\usepackage{amsmath}
\usepackage{amssymb}
\usepackage{amsfonts}
\usepackage{array}
\usepackage{tikz}
\usepackage{hhline}
\usepackage{mathtools}
\usepackage{multirow}
\usepackage{subfig}
\usepackage{hyperref}
\usepackage{mathtools}
\usepackage{tikz}
\usepackage{makecell}
\usepackage{pgfplots}
\usepackage{afterpage}
\usepackage{fancyhdr}
\usepackage{setspace}
\usepackage{algorithmic}

\newtheorem{lemma} {Lemma}

\newtheorem{theorem}{Theorem}

\newtheorem{conjecture}{Conjecture}

\newenvironment{proof}{\noindent{\bf Proof:}\indent}%
                      {\hfill $\Box$\par}

\newcommand{\sym}[1]{{\sf #1}}

\title{On The \say{Majority is Least stable} Conjecture} 
\author{Aniruddha Biswas and Palash Sarkar \\
Indian Statistical Institute \\
203, B.T.Road, Kolkata \\
India 700108. \\
Email: \{anib\_r, palash\}@isical.ac.in
}
\date{\today}

\begin{document}

\maketitle

\begin{abstract}
We show that the ``majority is least stable'' conjecture is true for $n=1$ and $3$ and false for all odd $n\geq 5$. 
\end{abstract}
\section{Introduction}
A Boolean function $f:\{-1,1\}^n\rightarrow \{-1,1\}$ is said to be a linear threshold function (LTF) if there are real constants $w_0,w_1,\ldots,w_n$ such that for 
any $\mathbf{x}=(x_1,\ldots,x_n)\in\{-1,1\}^n$, $f(\mathbf{x})=\sym{sgn}(w_0+w_1x_1+\cdots+w_nx_n)$, where $\sym{sgn}(z) = 1$ if $z \geq 0$, and $-1$ if $z < 0$. 

For $\mathbf{x}\in\{-1,1\}^n$ and $\rho\in[0,1]$, define a distribution $N_{\rho}(\mathbf{x})$ over $\{-1,1\}^n$ in the following manner:
$\mathbf{y}=(y_1,\ldots,y_n)\sim N_{\rho}(\mathbf{x})$ if for $i=1,\ldots,n$, $y_i=x_i$ with probability $\rho$ and $y_i=\pm1$ with probability $(1-\rho)/2$ each.
The noise stability of a function $f:\{-1,1\}^n\rightarrow\mathbb{R}$, denoted by $\sym{Stab}_{\rho}(f)$, is defined as follows.
\begin{equation*}
 \sym{Stab}_{\rho}(f)=\mathop{\mathbb{E}}_{\mathbf{x}\sim\{-1,1\}^n\text{, }\mathbf{y}\sim N_{\rho}(\mathbf{x})}\left[f(\mathbf{x})f(\mathbf{y})\right].
\end{equation*}
For odd $n$, the majority function $\sym{Maj}_n:\{-1,1\}^n \rightarrow \{-1,1\}$ is the following.
\begin{equation*}
 \sym{Maj}_n(x_1,\ldots,x_n)=\sym{sgn}(x_1 + x_2 + \ldots + x_n).
\end{equation*}

Benjamini, Kalai and Schramm in 1999 (see \cite{benjamini1999noise,filmus2014real}) put forward the following conjecture.
\begin{conjecture}\label{conjecture_1}
$\left(\text{\say{Majority is Least Stable}}\right):$ Let $n$ be odd and $f : \{-1, 1\}^n \rightarrow \{- 1, 1\}$ be an LTF. Then for all $\rho\in[0,1]$, $\sym{Stab}_{\rho}(f)\geq \sym{Stab}_{\rho}(\sym{Maj}_{n})$.
\end{conjecture}
A counterexample to the conjecture for $n=5$ has been reported in~\cite{jain2017counterexample} by Vishesh Jain where it is also mentioned that there are 
other known counterexamples to this conjecture by Sivakanth Gopi (2013), and Steven Heilman and Daniel Kane (2017). We could not locate these other counterexamples.
As of 2021, the conjecture is mentioned on Page~133 of the book on Boolean functions by O'Donnell~\cite{o2014analysis}.

In this note, we show that Conjecture~\ref{conjecture_1} is true for $n=1$ and 3 and false for odd $n\geq 5$. To show that the conjecture is false for odd $n\geq 5$, we define a
sequence of Boolean functions $g_n$ and show that $\sym{Stab}_{\rho}(g_n)<\sym{Stab}_{\rho}(\sym{Maj}_{n})$. To show that the conjecture is true for $n=3$, we employed a
search over all locally monotone 3-variable Boolean functions $f$ and obtained the expressions for $\sym{Stab}_{\rho}(f)$. It turns out that each of these
expressions is greater than or equal to $\sym{Stab}_{\rho}(\sym{Maj}_{n})$ for all $\rho\in [0,1]$.

\section{Preliminaries \label{sec-prelim} }
For a positive integer $n$, let $[n] = \{1, 2,\ldots, n\}$ and $2^{[n]}$ be the power set of $[n]$. 
The Fourier transform of $f:\{-1,1\}^n\rightarrow \{-1,1\}$ is a map $\hat{f}:2^{[n]}\rightarrow [-1,1]$ defined as follows. For $S\subseteq [n]$, 
\begin{eqnarray}\label{eqn-fourier}
	\hat{f}(S) & = & \frac{1}{2^n} \sum_{\mathbf{x}=(x_1,\ldots,x_n)\in\{-1,1\}^n} f(\mathbf{x}) \prod_{i\in S}x_i.
\end{eqnarray}

For $f:\{-1,1\}^n\rightarrow\{-1,1\}$ and $k\in\{0,\ldots,n\}$, let ${W}^{(k)}[f]=\sum_{S\subseteq[n],|S|=k}\hat{f}^2(S)$
and ${W}^{\leq k}[f]=\sum_{i=0}^k {W}^{(i)}[f]$.  
We say that $f$ is balanced if $\#\{\mathbf{x}:f(\mathbf{x})=1\} = \#\{\mathbf{x}:f(\mathbf{x})=-1\}$. It follows that $f$ is balanced if and only if 
$\hat{f}(\emptyset)=0$.

The Fourier expression of $\sym{Stab}_{\rho}(f)$ is the following (see Page~56 of~\cite{o2014analysis}). 
\begin{equation}\label{noise_stability}
    \sym{Stab}_{\rho}(f) = \sum_{k=0}^{n}\rho^{k}\cdot {W}^{(k)}[f].
\end{equation}
It is easy to see that $\sym{Maj}_n$ is balanced and so ${W}^{(0)}[\sym{Maj}_n]=0$. It is known that (see Page~62 of~\cite{o2014analysis}) 
\begin{equation}\label{maj_level_1_fourier_weight}
    {W}^{(1)}[\sym{Maj}_{n}]=\left[\frac{{n-1\choose \frac{n-1}{2}}}{2^{n-1}}\right]^2\cdot n.
\end{equation}
It was observed in~\cite{jain2017counterexample} that if $f$ is a balanced linear threshold function, then showing 
${W}^{(1)}[f] < {W}^{(1)}[\sym{Maj}_n]$ would disprove Conjecture~\ref{conjecture_1}. For the sake of completeness, we state a more general form of this observation 
as a lemma and provide a proof.

\begin{lemma}\label{our_lemma_1}
Let $n$ be odd and $f:\{-1,1\}^n\rightarrow\{-1,1\}$ be a Boolean function such that ${W}^{(0)}[f]=0$ and ${W}^{(1)}[f]<{W}^{(1)}[\sym{Maj}_n]$. 
Then there exists a $\delta>0$ such that $\sym{Stab}_{\rho}(f)<\sym{Stab}_{\rho}(\sym{Maj}_{n})$ for all $0<\rho<\delta$. Consequently, the function $f$ is a 
	counter-example to Conjecture~\ref{conjecture_1}.
\end{lemma}

\begin{proof}
	For $k\geq 0$, let $a_k={W}^{(k)}[f]-{W}^{(k)}[\sym{Maj}_n]$. Since by assumption, ${W}^{(0)}[f]=0$, 
	${W}^{(1)}[f]<{W}^{(1)}[\sym{Maj}_n]$, and noting that $\sym{Maj}_n$ is balanced, it follows that $a_0=0$ and $-1\leq a_1<0$. On the
	other hand, for $k\geq 2$, we have $-1\leq a_k<1$. 

	Now, $\sym{Stab}_{\rho}(f)-\sym{Stab}_{\rho}(\sym{Maj}_n)=\sum_{k=1}^{n}\rho^{k}\cdot a_k$.
	Therefore, $\sym{Stab}_{\rho}(f)-\sym{Stab}_{\rho}(\sym{Maj}_n)<0$ if and only if $\rho(a_2+\rho a_3+\ldots+\rho^{n-2}a_n)<-a_1$. Since $a_k<1$ for $k=2,\ldots,n$,
	it follows that 
	$\rho(a_2+\rho a_3+\ldots+\rho^{n-2}a_n)$ is upper bounded by $\rho(1+\rho+\ldots+\rho^{n-2})$ whose limiting value is $0$ as $\rho\to0$. Therefore, there must 
	exist some $\delta>0$ such that for all $0<\rho<\delta$, $\rho(a_2+\rho a_3+\ldots+\rho^{n-2}a_n)<-a_1$. Consequently, 
	$\sym{Stab}_{\rho}(f)<\sym{Stab}_{\rho}(\sym{Maj}_{n})$ for all $0<\rho<\delta$.
\end{proof}

Next we introduce the notion of influence of a variable on a Boolean function. For $i\in[n]$ and 
$f: \{-1,1\}^n \rightarrow \{-1,1\}$, $\sym{Inf}_i(f)$ is defined as follows (see Page~46 of~\cite{o2014analysis}).
\begin{equation*}
    \sym{Inf}_{i}(f) = \Pr_{\mathbf{x} \in \{-1,1\}^n}[f(\mathbf{x}) \neq f(\mathbf{x^{\oplus i}})],
\end{equation*}
where $\mathbf{x}^{\oplus i}$ denotes the vector $(x_1,\ldots , x_{i-1}, -x_i , x_{i+1}, \ldots , x_n)$.

An $n$-variable Boolean function $f$ is said to be locally monotone if it is monotone increasing or decreasing in each variable.
From~\cite{gotsman1994spectral} (see Lemma~2.2 and the comment following it), it follows that if $f$ is a locally monotone function, then for all $i\in [n]$, 
$\sym{Inf}_i(f)=|\hat{f}(\{i\})|$. Since an LTF is locally monotone, we have the following result which has been used in the proof of Theorem~4.1 of~\cite{gotsman1994spectral}.
\begin{theorem} \label{gotsman_linial_theorem}\cite{gotsman1994spectral}
If $f:\{-1,1\}^n\rightarrow\{-1,1\}$ is an LTF. Then $\sum_{i=1}^n\sym{Inf}_i(f)^2 = {W}^{(1)}[f]$.
\end{theorem}

\section{Settling Conjecture~\ref{conjecture_1} \label{sec-conj1}}
We state and prove some results from which the main theorem follows.

\begin{lemma}\label{lemma_n_odd}
	Let $n\geq 1$, $w_0$ be an integer and $w_1$ and $w_2$ be non-zero integers. Let $T$ be a subset of $[n]$ of cardinality $t\leq n/2$.
	Consider the following LTF: $$f(x_1,\ldots,x_n)=\sym{sgn}\left(w_0+w_1\cdot\sum_{u\in T}x_{u}+w_2\cdot\sum_{v\in \overline{T}}x_{v}\right).$$ Then

\begin{eqnarray}
    	{W}^{(0)}[f] 
	& = & \frac{1}{2^{2n}} \cdot \left[\sum_{(i,j)\in \mathcal{S}_0}{t \choose i}{n-t\choose j}\right]^2, \label{eq_n_odd0} \\
	{W}^{(1)}[f] 
	& = & \frac{t}{2^{2n-2}} \cdot \left[\sum_{(i,j)\in \mathcal{S}_1}{t-1 \choose i}{n-t\choose j}\right]^2 
			+ \frac{n-t}{2^{2n-2}} \cdot \left[\sum_{(i,j)\in \mathcal{S}_2}{t \choose i}{n-t-1\choose j}\right]^2, \label{eq_n_odd1} 
\end{eqnarray}
where $\mathcal{S}_{0}$,  $\mathcal{S}_{1}$ and $\mathcal{S}_{2}$ are defined as follows.
\begin{eqnarray*}
	\begin{array}{l}
		\mathcal{S}_{0}=
		\left\{
			\begin{array}{ll}
				\{(i,j): 0\leq i \leq t \text{, } 0\leq j \leq n-t \\
				\qquad\qquad \text{ and } -w_0\leq w_1(2i-t)+w_2(2j-(n-t)) \leq w_0\} & \text{if } w_0\geq 0, \\
				\{(i,j): 0\leq i \leq t \text{, } 0\leq j \leq n-t \\
				\qquad \qquad \text{ and } w_0< w_1(2i-t)+w_2(2j-(n-t)) < -w_0\} & \text{if } w_0<0;
			\end{array}

		\right. \\
		\mathcal{S}_{1} = \{(i,j): 0\leq i \leq t-1 \text{, } 0\leq j \leq n-t \\ 
		\qquad\qquad \text{ and } -|w_1|\leq w_0+w_1(2i-(t-1))+w_2(2j-(n-t))<|w_1|\}; \\
		\mathcal{S}_{2} = \{(i,j): 0\leq i \leq t \text{, } 0\leq j \leq n-t-1 \\
		\qquad \qquad \text{ and } -|w_2|\leq w_0+w_1(2i-t)+w_2(2j-(n-t-1)) < |w_2|\}.
	\end{array}
\end{eqnarray*}
	
\end{lemma}

\begin{proof}
	We start with the proof of~\eqref{eq_n_odd0}.
	For $\mathbf{x}\in\{-1,1\}^n$, let $A(\mathbf{x})=w_1\cdot\sum_{u\in T}x_{u}+w_2\cdot\sum_{v\in \overline{T}}x_{v}$ so that
	$f(\mathbf{x})=\sym{sgn}(w_0+A(\mathbf{x}))$. Let $N$ (resp. $M$) be the number of $\mathbf{x}$'s such that 
	$w_0+A(\mathbf{x})\geq 0$ (resp. $w_0+A(\mathbf{x})<0$). Then $\hat{f}(\emptyset)=(N-M)/2^n$. There are two cases to consider.

	First consider the case $w_0\geq0$. Let $N_1$ (resp. $N_2$) be the number of $\mathbf{x}$'s such that 
	$A(\mathbf{x})>-w_0$ (resp. $-w_0\leq A(\mathbf{x})\leq w_0$). So, $N=N_1+N_2$. Since $A(-\mathbf{x})=-A(\mathbf{x})$, it follows that $N_1=M$ and so
	$\hat{f}(\emptyset)=N_2/2^n$. Therefore to obtain ${W}^{(0)}[f]=\hat{f}^2(\emptyset)$ it is sufficient to obtain $N_2$.
	For $\mathbf{x}\in\{-1,1\}^n$, let $i=\#\{u\in T:x_u=1\}$ and $j=\#\{v\in\overline{T}:x_v=1\}$. Then $A(\mathbf{x})=w_1(2i-t)+w_2(2j-(n-t))$. 
	For $0\leq i\leq t$ and $0\leq j\leq n-t$, the pair $(i,j)$ is in $\mathcal{S}_0$ if and only if $-w_0\leq A(\mathbf{x})\leq w_0$. So, the number of $\mathbf{x}$'s
	for which $-w_0\leq A(\mathbf{x})\leq w_0$ holds is $\sum_{(i,j)\in\mathcal{S}_0}{t\choose i}{n-t\choose j}$ which is the value of $N_2$.

	    Next consider the case $w_0<0$. Let $M_1$ (resp. $M_2$) be the number of $\mathbf{x}$'s such that 
	$A(\mathbf{x})\leq w_0$ (resp. $w_0< A(\mathbf{x})< -w_0$). So, $M=M_1+M_2$. Again since $A(-\mathbf{x})=-A(\mathbf{x})$, it follows that $M_1=N$ and so
	$\hat{f}(\emptyset)=-M_2/2^n$. Therefore to obtain ${W}^{(0)}[f]=\hat{f}^2(\emptyset)$ it is sufficient to obtain $M_2$. 
	A similar argument as above shows that $M_2$ is equal to $\sum_{(i,j)\in\mathcal{S}_{0}}{t \choose i}{n-t\choose j}$. 

	Now we turn to the proof of~\eqref{eq_n_odd1}. 
	Fix some $s\in T$ and some $r\in \overline{T}$. Due to symmetry, for any $i\in T$, we have $\sym{Inf}_i(f)=\sym{Inf}_s(f)$ and for any $j\in \overline{T}$, 
	we have $\sym{Inf}_j(f)=\sym{Inf}_r(f)$ and so from Theorem~\ref{gotsman_linial_theorem}, 
\begin{equation}\label{local_eq}
	{W}^{(1)}[f]=t\cdot \sym{Inf}_{s}(f)^2 + (n-t)\cdot \sym{Inf}_{r}(f)^2.
\end{equation}
	Let $N_s$ (resp. $N_r$) be the number of $\mathbf{x}\in \{-1,1\}$ such that $f(\mathbf{x})\neq f(\mathbf{x}^{\oplus s})$ (resp.
	$f(\mathbf{x})\neq f(\mathbf{x}^{\oplus r})$). Then $\sym{Inf}_s(f)=N_s/2^{n-1}$ and $\sym{Inf}_r(f)=N_r/2^{n-1}$. 
	
	For $\mathbf{x}\in\{-1,1\}^n$, let $B(\mathbf{x})=w_0+ w_1 \sum_{u\in T\setminus\{s\}}x_u+ w_2\sum_{v\in \overline{T}}x_v$. From the definition of $f$, 
	$N_s$ is the number of $\mathbf{x}$'s such that either ($w_1x_s+B(\mathbf{x}) \geq 0$ and $-w_1x_s+B(\mathbf{x})<0$) or ($w_1x_s+B(\mathbf{x}) < 0$ and 
	$-w_1x_s+B(\mathbf{x})\geq0$) holds. The two conditions are equivalent to 
	$-w_1x_s\leq B(\mathbf{x}) < w_1x_s$ and $w_1x_s\leq B(\mathbf{x}) < -w_1x_s$ respectively. For the first condition, we must have $w_1x_s>0$ as
	otherwise we obtain $|w_1x_s|\leq B(\mathbf{x})<-|w_1x_s|$ which is a contradiction since $w_1x_s$ is non-zero; similarly, for the 
	second condition, we must have $w_1x_s<0$. Noting that $x_s\in\{-1,1\}$, both the conditions 
	boil down to $-|w_1|\leq B(\mathbf{x})<|w_1|$, and consequently, $N_s$ is the number of $\mathbf{x}$'s such that $-|w_1|\leq B(\mathbf{x}) <|w_1|$ holds.

	For $\mathbf{x}\in \{-1,1\}^n$, let $i=\#\{u\in T\setminus\{s\}:x_u=1\}$ and $j=\#\{v\in \overline{T}:x_v=1\}$. Then $B(\mathbf{x})=w_0+w_1(2i-(t-1))+w_2(2j-(n-t))$. 
	For $0\leq i\leq t-1$ and $0\leq j\leq n-t$, the pair $(i,j)$ is in $\mathcal{S}_1$ if and only if $-|w_1|\leq B(\mathbf{x})<|w_1|$ holds. So, the number of $\mathbf{x}$'s 
	for which $-|w_1|\leq B(\mathbf{x}) <|w_1|$ holds is $\sum_{(i,j)\in \mathcal{S}_1} {t-1\choose i}{n-t\choose j}$ which is the value of $N_s$. 
	
	A similar argument shows that $N_r$ is equal to $\sum_{(i,j)\in\mathcal{S}_{2}}{t \choose i}{n-t-1\choose j}$. 
	Using the values of $N_s$ and $N_r$ to obtain $\sym{Inf}_s(f)$ and $\sym{Inf}_r(f)$ respectively and substituting these in~\eqref{local_eq} gives
	the expression for ${W}^{(1)}[f]$ stated in~\eqref{eq_n_odd1}.
\end{proof}

For odd $n\geq 3$, we define a sequence of functions $g_n:\{-1,1\}^n\rightarrow\{-1,1\}$ where
\begin{eqnarray}\label{eqn-g}
	g_{n}(x_1,\ldots,x_n) & = & \sym{sgn}(2\cdot(x_1+\ldots+x_{n-3})+x_{n-2}+x_{n-1}+x_n).
\end{eqnarray}
In~\cite{jain2017counterexample}, the function $g_5$ has been shown to be a counter-example to Conjecture~\ref{conjecture_1}.
\begin{lemma}\label{our_lemma_1a}
	For $g_n$ defined in~\eqref{eqn-g}, we have
	\begin{eqnarray}
	{W}^{(0)}[g_n] & = & 0, \nonumber \\
	{W}^{(1)}[g_{n}] 
		& = & (n-3)\cdot \left[\frac{{n-4\choose \frac{n-5}{2}}\cdot 8}{2^{n-1}}\right]^2 + 3\cdot \left[\frac{{n-3\choose \frac{n-3}{2}}\cdot 2}{2^{n-1}}\right]^2.
				\label{eqn-W1g}
	\end{eqnarray}
\end{lemma}

\begin{proof} 
	We use Lemma~\ref{lemma_n_odd}. For $g_n$, we have $w_0=0$, $w_1=1$ and $w_2=2$. Also, $T=\{n-2,n-1,n\}$ so that $t=3$. With these values, the sets $\mathcal{S}_0$, 
	$\mathcal{S}_1$ and $\mathcal{S}_2$ defined in Lemma~\ref{lemma_n_odd} are the following.
	\begin{eqnarray*}
		\mathcal{S}_{0} & = & \{(i,j): 0\leq i \leq 3 \text{, } 0\leq j \leq n-3 \text{ and } (2i-3)+2(2j-(n-3))=0\}, \\
		\mathcal{S}_{1} & = & \{(i,j): 0\leq i \leq 2 \text{, } 0\leq j \leq n-3 \text{ and } -1\leq (2i-2)+2(2j-(n-3))<1\}, \\
		\mathcal{S}_{2} & = & \{(i,j): 0\leq i \leq 3 \text{, } 0\leq j \leq n-4 \text{ and } -2\leq (2i-3)+2(2j-(n-4)) < 2\}.
	\end{eqnarray*}
	Since $(2i-3)+2(2j-(n-3))$ is odd, it cannot be zero and so $\mathcal{S}_0$ is empty showing that ${W}^{(0)}[g_n]=0$.

	Since $(2i-2)+2(2j-(n-3))$ is even it cannot be equal to $-1$ and so the only possible value it can take is 0. From this, we obtain 
	$\mathcal{S}_1$ to be $\{(1,(n-3)/2)\}$. 
	
	Similarly, since $(2i-3)+2(2j-(n-4))$ is odd, the only possible values in the set $\{-2,-1,0,1\}$ that it can take are $-1$ and $1$. Corresponding to these 
	two values, we obtain $j=(n-3-i)/2$ and $j=(n-2-i)/2$ respectively. Since $n$ is odd, in the first case $i$ must be even, while in the second case $i$ must
	be odd. So, $\mathcal{S}_2=\{(0,(n-3)/2),(2,(n-5)/2),(1,(n-3)/2),(3,(n-5)/2)\}$.

	Substituting the values of $w_0$, $w_1$, $w_2$, $t$ as well as $\mathcal{S}_1$ and $\mathcal{S}_2$ in~\eqref{eq_n_odd1}, we obtain 

	\begin{eqnarray*}
{W}^{(1)}[g_n]
		&=& \frac{3}{2^{2n-2}}\cdot \left[{2 \choose 1}{n-3\choose \frac{n-3}{2}}\right]^2 
		+\frac{n-3}{2^{2n-2}}\cdot \left[{n-4\choose \frac{n-3}{2}}+3 {n-4\choose \frac{n-5}{2}}+3{n-4\choose \frac{n-3}{2}} +{n-4\choose \frac{n-5}{2}}\right]^2 \\
	\end{eqnarray*}
	Noting that $(n-3)/2+(n-5)/2=n-4$ leads to the expression for ${W}^{(1)}[g_n]$ given in~\eqref{eqn-W1g}.
\end{proof}

\begin{lemma}\label{our_lemma_2}
	Let $g_n$ be defined as in~\eqref{eqn-g}. For odd $n\geq 5$, ${W}^{(1)}[g_{n}]< {W}^{(1)}[\sym{Maj}_n]$.
\end{lemma}
\begin{proof}
	The expression for ${W}^{(1)}[\sym{Maj}_{n}]$ is given by~\eqref{maj_level_1_fourier_weight} and the expression for
	${W}^{(1)}[g_{n}]$ is given by~\eqref{eqn-W1g}. Therefore
\begin{flalign}\label{compare_g1_maj}
\quad\quad\frac{{W}^{(1)}[g_{n}]}{{W}^{(1)}[\sym{Maj}_{n}]}&= \left[\frac{{n-4 \choose \frac{n-5}{2}}\cdot 8}{{n-1 \choose \frac{n-1}{2}}}\right]^2\left(\frac{n-3}{n}\right) + \left[\frac{{n-3 \choose \frac{n-3}{2}}\cdot 2}{{n-1 \choose \frac{n-1}{2}}}\right]^2\left(\frac{3}{n}\right) 
	=\left[\frac{n-1}{n-2}\right]^2\left(\frac{4n-9}{4n}\right).
\end{flalign}
	From~\eqref{compare_g1_maj}, it follows that ${W}^{(1)}[g_{n}]< {W}^{(1)}[\sym{Maj}_n]$ if and only if $(n-3)^2>0$ i.e. $n\geq5$.
\end{proof}
Note that from~\eqref{compare_g1_maj}, for $n=3$ we have ${W}^{(1)}[g_{3}]={W}^{(1)}[\sym{Maj}_3]$.

\begin{lemma}\label{our_lemma_3} 
	Conjecture~\ref{conjecture_1} is true for $n=1$ and $n=3$.
\end{lemma}

\begin{proof}
	For $n=1$, the only LTF is the majority function and so Conjecture~\ref{conjecture_1} is trivially true. 

	Using~\eqref{noise_stability}, for $n=3$, it is easy to check that $\sym{Stab}_{\rho}(\sym{Maj}_{3})=0.75\rho+ 0.25\rho^3$.
	We need to compare this expression with $\sym{Stab}_{\rho}(f)$ where $f$ is an LTF. 
	We used an exhaustive search. There is no easy way to determine whether a given function is an LTF. Instead we considered the set of all 3-variable locally monotone 
	functions. Since an LTF is locally monotone, our search covered all LTFs.
	Let $f$ be a $3$-variable locally monotone function. We obtained the values of ${W}^{(k)}[f]$, for $k=0,1,2,3$ and using~\eqref{noise_stability} obtained
	the expression for $\sym{Stab}_{\rho}(f)$. From the search, the possible expressions for $\sym{Stab}_{\rho}(f)$ were obtained to be the following: 
$1$, $\rho$, $0.75\rho+ 0.25\rho^3$, 
$0.0625+0.6875\rho+0.1875\rho^2+0.0625\rho^3$, 
$0.25+0.5\rho+0.25\rho^2$, 
$0.5625+0.1875\rho+0.1875\rho^2+0.0625\rho^3$. 
	For each of these expressions, it is easy to verify that $\sym{Stab}_{\rho}(f)\geq\sym{Stab}_{\rho}(\sym{Maj}_{3})$ for all $\rho\in [0,1]$.
\end{proof}

Based on Lemmas \ref{our_lemma_1},~\ref{our_lemma_2} and~\ref{our_lemma_3}, we obtain the main result of the paper, of which the case $n=5$ was 
reported in~\cite{jain2017counterexample}.
\begin{theorem}\label{main_theorem}
Conjecture~\ref{conjecture_1} is true for $n=1$ and $n=3$. For odd $n\geq 5$, Conjecture~\ref{conjecture_1} is false.
\end{theorem}

For $n\geq 5$, there are other functions which provide counterexamples to Conjecture~\ref{conjecture_1}.
For odd $n$, suppose $h_n$ is defined as $h_n(\mathbf{x})=\sym{sgn}(2\cdot(x_1+\ldots+x_{n-3})-(x_{n-2}+x_{n-1}+x_n))$. Then proceeding as in the proof of Lemma~\ref{our_lemma_1a},
it is possible to show that ${W}^{(0)}[h_n] =0$ and ${W}^{(1)}[h_{n}]={W}^{(1)}[g_{n}]$. So, for odd $n\geq 5$, the function $h_n$ is a counterexample to
Conjecture~\ref{conjecture_1}. 

Futher, for concrete values of $n$, it is possible to obtain examples of functions $f_n$ such that $W^{\leq 1}[f_n] < W^{\leq 1}[g_n]$. Suppose $t$ is odd and
$n=t^2$. For positive integer $w$, define the function $f_n^{(w)}$ in the following manner.
\begin{eqnarray*}
	f_n^{(w)}(x_1,\ldots,x_n)
	& = & \left\{
		\begin{array}{ll}
			\sym{sgn}((2w+t-1)+2w(x_1+\cdots+x_{t})+ \\
			\qquad \cdots+(2w+t-1)(x_{n-t+1}+\cdots+x_n)) & \text{if $(t+1)/2$ is even}, \\
			\sym{sgn}((2w+t)+(2w+1)(x_1+\cdots+x_{t})+\\
			\qquad \cdots+(2w+t)(x_{n-t+1}+\cdots+x_n)) & \text{if $(t+1)/2$ is odd.}
		\end{array}
		\right.
\end{eqnarray*}
Computations show that
$W^{\leq 1}[f_9^{(4)}] = 0.651764 < 0.659180 = W^{\leq 1}[g_9]$ and $W^{\leq 1}[f_{25}^{(12)}] = 0.640686 < 0.643535 = W^{\leq 1}[g_{25}]$.

%
%

\section{Limiting Value of ${W}^{\leq 1}[g_n]$}
It is known (see Page~62 of~\cite{o2014analysis}) that ${W}^{(\leq 1)}[\sym{Maj}_n]$ is a decreasing sequence which is lower bounded by $2/\pi$.
It has been conjectured (see~\cite{benjamini1999noise} and Page~115 of~\cite{o2014analysis}) that if $f$ is an LTF, then ${W}^{\leq 1}[f] \geq \frac{2}{\pi}$.

We have shown that for odd $n\geq 5$, the function $g_n$ defined by~\eqref{eqn-g} satisfies ${W}^{(1)}[g_n]<{W}^{(1)}[\sym{Maj}_n]$. This brings
up the question of whether the sequence $g_n$ also provides a counter-example to the $2/\pi$ lower bound conjecture for LTFs. In this section, we show that this is not the
case.

From Lemma~\ref{our_lemma_1a}, ${W}^{(0)}[g_n]=0$ and so, ${W}^{\leq 1}[g_n]={W}^{(1)}[g_n]$. This shows that it is sufficient to 
consider ${W}^{(1)}[g_n]$.
We show that ${W}^{(1)}[g_n]$ is a decreasing sequence which is lower bounded by $2/\pi$. The expression for ${W}^{(1)}[g_n]$ given 
by~\eqref{eqn-W1g} involves binomial coefficients. We use the following bounds on factorial function (see Page~54 of~\cite{feller1968introduction}).
\begin{eqnarray}\label{eqn-str}
	\sqrt{2\pi m}\cdot \frac{m^m}{e^m}\exp\left(\frac{1}{12m+1}\right) \leq m! \leq \sqrt{2\pi m}\cdot\frac{m^m}{e^m}\exp\left(\frac{1}{12m}\right). 
\end{eqnarray}
Let $p=\frac{k}{m}$ and $q=1-p$. Using~\eqref{eqn-str}, the following bounds on ${m \choose k}$ can be obtained. 
\begin{eqnarray} \label{eqn-bin}
	\left. \begin{array}{rcl}
	{m \choose k} & \geq & \frac{1}{\sqrt{2\pi mpq}}(p^pq^q)^{-m}\exp\left(\frac{1}{12m+1}-\frac{1}{12k}-\frac{1}{12(m-k)}\right), \\ 
	{m \choose k} & \leq & \frac{1}{\sqrt{2\pi mpq}}(p^pq^q)^{-m}\exp\left(\frac{1}{12m}-\frac{1}{12k+1}-\frac{1}{12(m-k)+1}\right). 
	\end{array} \right\}
\end{eqnarray}

\begin{lemma}\label{our_lemma_4}
	For $g_n$ defined in~\eqref{eqn-g}, ${W}^{(1)}[g_n]$ is a decreasing sequence and $\lim_{n\to\infty}{W}^{(1)}[g_n]=\frac{2}{\pi}$.
	Consequently, for all odd $n$, ${W}^{(1)}[g_n]\geq 2/\pi$.
\end{lemma}
\begin{proof}
	Let $a_n={W}^{(1)}[g_{n}]$ and $b_n={W}^{(1)}[\sym{Maj}_n]$. We wish to show that $a_n$ is a decreasing sequence. 
	To do this, we compare $a_{n+2}/b_n$ to $a_n/b_n$. The expression for $a_n/b_n$ is given by~\eqref{compare_g1_maj}. Using~\eqref{maj_level_1_fourier_weight} 
	and~\eqref{eqn-W1g}, we obtain $a_{n+2}/b_n=(4n-1)/(4n)$. We have $a_n \geq a_{n+2}$ if and only if $a_n/b_n\geq a_{n+2}/b_n$. 
	Using the expressions for $a_n/b_n$ and $a_{n+2}/b_n$, the last condition is equivalent to 
$$	
\left[\frac{n-1}{n-2}\right]^2\left(\frac{4n-9}{4n}\right) > \frac{4n-1}{4n}
$$
which holds if and only if $n\geq 3$. So, $a_n$ is a decreasing sequence for all odd $n\geq 3$.

Let $A_n=(n-3)\cdot\left[\frac{{n-4\choose \frac{n-5}{2}}\cdot 8 }{2^{n-1}}\right]^2$ and $B_n=3\cdot\left[\frac{{n-3\choose \frac{n-3}{2}}\cdot 2}{2^{n-1}}\right]^2$ 
and so $a_n=A_n+B_n$. We show that $A_n$ tends to $2/\pi$ and $B_n$ tends to $0$ and so $a_n$ tends to $2/\pi$ as $n$ goes to infinity.
	
	First consider $A_n$. Letting $m=n-4$, $k=\frac{n-5}{2}$, $p=k/m$ and $q=1-p$, from~\eqref{eqn-bin} and using some routine simplifications
	we obtain the following bounds on $A_n$. 
\begin{eqnarray*}
	A_n & \geq & \frac{2}{\pi}\left[\frac{n-4}{n-3}\right]^{(n-3)}\left[\frac{n-5}{n-4}\right]^{-(n-4)}\exp\left(\frac{2}{12n-47}-\frac{2}{6n-30}-\frac{2}{6n-18}\right), \\
	A_n & \leq & \frac{2}{\pi}\left[\frac{n-4}{n-3}\right]^{(n-3)}\left[\frac{n-5}{n-4}\right]^{-(n-4)}\exp\left(\frac{2}{12n-48}-\frac{2}{6n-29}-\frac{2}{6n-17}\right).
\end{eqnarray*}
Since $\lim_{x\to\infty}(1+\frac{1}{x})^x=e$ and $\lim_{x\to\infty}(1-\frac{1}{x})^x=\frac{1}{e}$, it follows that $\lim_{n\to\infty}A_n=2/\pi$.



Now, consider $B_n$. Letting $m=n-3$, $k=\frac{n-3}{2}$ and $p=q=\frac{1}{2}$, from~\eqref{eqn-bin} and using some routine simplifications
        we obtain the following bounds on $B_n$.
\begin{eqnarray*}
	B_n & \geq & \frac{3}{2\pi}\left(\frac{1}{n-3}\right)\exp\left(\frac{2}{12n-35}-\frac{2}{6n-18}-\frac{2}{6n-18}\right), \\
	B_n & \leq & \frac{3}{2\pi}\left(\frac{1}{n-3}\right)\exp\left(\frac{2}{12n-36}-\frac{2}{6n-17}-\frac{2}{6n-17}\right).
\end{eqnarray*}
It follows that $\lim_{n\to\infty}B_n=0$.

\end{proof}

\bibliographystyle{ltf}
\bibliography{ltf.bib}
\end{document}